\documentclass[11pt]{article}
\usepackage[utf8]{inputenc}
\usepackage{latexsym,amssymb,amsmath,graphics,cite,arydshln}
\usepackage{tikz}
\usepackage{fancyhdr}
\usepackage{lipsum}
\usepackage{lineno}
\usepackage{color}
\usepackage{makecell}
\usepackage{amsfonts}
\usepackage{multirow} %multirow for format of table  
\usepackage{xcolor}
\usepackage{colortbl}
\usepackage{booktabs}
\usepackage{diagbox}

\usepackage[colorlinks,
linkcolor=blue,       %%修改此处为你想要的颜色
anchorcolor=blue,  %%修改此处为你想要的颜色
citecolor=blue,        %%修改此处为你想要的颜色，例如修改blue为red
]{hyperref}

\newtheorem{Theorem}{Theorem}[section]

\newtheorem{Lemma}[Theorem]{Lemma}

\newtheorem{Construction}[Theorem]{Construction}
\newtheorem{Remark}[Theorem]{Remark}
\newcommand{\qed}{\hphantom{.}\hfill $\Box$\medbreak}

\begin{document}
\title{New constructions of cyclic subspace codes\footnote{Supported by the National Natural Science Foundation of China under Grant 12271390 (Ji). (Corresponding author: Lijun Ji)}}

\author{\small  Shuhui \ Yu, Lijun\ Ji \\	\small   Department of Mathematics, Soochow University,	Suzhou 215006, China\\	\small E-mail address: yushuhui\_suda@163.com, jilijun@suda.edu.cn}
\date{}
\maketitle
\begin{abstract}
A subspace of a finite field is called a Sidon space if the product of any two of its nonzero elements is unique up to a scalar multiplier from the base field. Sidon spaces, introduced by  Roth et al. (IEEE Trans Inf Theory 64(6): 4412-4422, 2018), have a close connection with optimal full-length orbit codes. In this paper, we present two constructions of Sidon spaces. The union of Sidon spaces from the first construction yields cyclic subspace codes in $\mathcal{G}_{q}(n,k)$ with minimum distance $2k-2$ and size $r(\lceil \frac{n}{2rk} \rceil -1)((q^{k}-1)^{r}(q^{n}-1)+\frac{(q^{k}-1)^{r-1}(q^{n}-1)}{q-1})$, where $k|n$, $r\geq 2$ and $n\geq (2r+1)k$, $\mathcal{G}_{q}(n,k)$ is the set of all $k$-dimensional subspaces of $\mathbb{F}_{q}^{n}$. The union of Sidon spaces from the second construction gives cyclic subspace codes in $\mathcal{G}_{q}(n,k)$ with minimum distance $2k-2$ and size $\lfloor \frac{(r-1)(q^{k}-2)(q^{k}-1)^{r-1}(q^{n}-1)}{2}\rfloor$ where $n= 2rk$ and $r\geq 2$. Our cyclic subspace codes have larger sizes than those in the literature, in particular, in the case of $n=4k$, the size of our resulting code is within a factor of $\frac{1}{2}+o_{k}(1)$ of the sphere-packing bound as $k$ goes to infinity. 

\medskip\noindent \textbf{Keywords}: Cyclic subspace code,  Sidon space,  Network coding\smallskip
\end{abstract}

\section{Introduction}
Let $q$ be a prime power, and $\mathbb{F}_{q}$ the finite field of size $q$. Let $\mathbb{F}_{q^{n}}$ be the extension field of degree $n$ over $\mathbb{F}_{q}$, which is also a vector space of dimension $n$ over $\mathbb{F}_{q}$.  The projective space of order $n$ over $\mathbb{F}_{q}$, denoted by $\mathcal{P}_{q}(n)$, is the set of all $\mathbb{F}_{q}$-subspaces of $\mathbb{F}_{q^{n}}$. For nonnegative integers $k\leq n$, we denote by $\mathcal{G}_{q}(n,k)$ the set of all $k$-dimensional $\mathbb{F}_{q}$-subspaces of $\mathbb{F}_{q^{n}}$. For any $U,V\in \mathcal{P}_{q}(n)$, the subspace distance of $U, V$  is defined as: $$d(U,V)=\dim(U)+\dim(V)-2\dim(U\cap V),$$ where $\dim(\cdot)$ denotes the dimension of a vector space over $\mathbb{F}_{q}$. A nonempty subset $\mathcal{C}$ of $\mathcal{P}_{q}(n)$ with the subspace distance is called a subspace code. The minimum subspace distance of a subspace code $\mathcal{C}$ is defined as $d(\mathcal{C})=\min\{d(U,V)\colon U, V\in \mathcal{C}, U\neq V\}.$ In particular, a subspace code $\mathcal{C}$ is termed a constant dimension code (CDC) if all the elements of $\mathcal{C}$ have the same dimension $k$. For convenience, a CDC in $\mathcal{G}_{q}(n,k)$ with minimum subspace distance $d$ is denoted by $(n,d,k)_q$-CDC. 

Subspace codes have  attracted much attension in recent decades due to their applications in random network coding, cf. \cite{ACLY2000,ES2009,EV2011,KK2008,KoK2008}. A central problem of subspace codes is to construct CDCs with large minimum distance and as many codewords as possible. 

 We recall the sphere-packing bound for CDCs. 

\begin{Theorem}[Sphere-Packing Bound \cite{RRT2017}]\label{UB}
Let $\mathcal{C}\subseteq \mathcal{G}_{q}(n,k)$ be a CDC with minimum distance $d$. Then $$|\mathcal{C}|\leq  \frac{{\small\left[ \begin{array}{ccccc}  n\\ k-\frac{d}{2}+1 \end{array}\right]_{q}}}{{\small \left[ \begin{array}{ccccc}  k\\ k-\frac{d}{2}+1 \end{array}\right]_{q}}},$$ where ${\small \left[ \begin{array}{ccccc} t\\ s \end{array}\right]_{q}}=\prod\limits_{i=0}^{s-1}\frac{q^{t-i}-1}{q^{i+1}-1}.$
\end{Theorem}

For a subspace $U\in \mathcal{G}_{q}(n,k)$ and $\alpha\in \mathbb{F}_{q^{n}}^{*}=\mathbb{F}_{q^{n}}\setminus \{0\}$, the cyclic shift of $U$ by $\alpha$ is $\alpha U=\{\alpha u\colon u\in U\}$. The orbit of $U$ is $orb(U)=\{\alpha U\colon \alpha \in \mathbb{F}_{q^{n}}^{*}\}$, each $\alpha U\in orb(U)$ is a subspace of dimension $k$. Clearly, $|orb(U)|=\frac{q^{n}-1}{q^{t}-1}$ for some $t | n$. If $|orb(U)|=\frac{q^{n}-1}{q-1}$, it is a full-length  orbit. If $d(orb(U))=2k-2$, then this full-length  orbit code $orb(U)$ is called optimal.
A subspace code $\mathcal{C}$ is called cyclic if $\alpha U\in \mathcal{C}$ for all $\alpha\in \mathbb{F}_{q^{n}}^{*}$ and $U\in \mathcal{C}$. A lot of efforts have been devoted to constructions of large cyclic subspace codes, cf\cite{BH2019,ETGR2016,RRT2017,NYW2020,FW2021,HC2021,YH2021,NXG2022, GMT2014,AFMJ2013,HC2022,TG2022,OO2017,WX2019,ZT2023}.  Cyclic CDCs are very useful with efficient encoding and decoding algorithms, cf. \cite{BEOV2016, GPB2010, KoK2008, EV2011}. There has been extensive reserch on constructions of $(n,2k-2,k)_q$-CDCs by using linearized polynomials, cf.\cite{BH2019,ETGR2016,OO2017,WX2019}. They used the set of roots of the linearized polynomial (with few monomials) to denote the subspace $\mathcal{U}$, and showed that $\{\alpha \mathcal{U}\colon \alpha \in \mathbb{F}_{q^{n}}^{*}\}$ is an $(n,2k-2,k)_q$-CDC for any given $k$ and infinitely many values of $n$. They also provided a method to increase the number of distinct cyclic subspace code without decreasing the minimum subspace distance.  Such constructions give cyclic subspace codes with large sizes, but it is hard to give an estimate of the largest possible sizes due to complicated conditions. 

We denote $a\cdot \mathbb{F}_{q}=\{\lambda a\colon \lambda \in \mathbb{F}_{q}\}$. A Sidon space is a subspace $U\in \mathcal{G}_{q}(n,k)$ such that for all nonzero elements $a$, $b$, $c$, $d\in U$, if $ab=cd$ then $\{a\mathbb{F}_{q}, b\mathbb{F}_{q}\}=\{c\mathbb{F}_{q}, d\mathbb{F}_{q}\}.$ In \cite{RRT2017}, Roth et al. showed that Sidon spaces have a close connection with optimal full-length orbit codes.

\begin{Lemma}[\cite{RRT2017}]\label{sac}
Let the subspace $U\in \mathcal{G}_{q}(n,k)$. The set $\mathcal{C}=\{\alpha U\colon \alpha\in \mathbb{F}_{q^{n}}^{*}\}$ is an optimal full-length orbit code if and only if $U$ is a Sidon space.
\end{Lemma}

In the paper \cite{RRT2017},  Roth et al. provided a construction of $(2k,2k-2,k)_q$-CDCs by taking the union of optimal full-length orbit codes arising from Sidon spaces. Such codes have sizes differ from the sphere-packing bound by a factor of $\frac{1}{2}$ asymptotically as $k$ goes to infinity.

\begin{Lemma}[\cite{RRT2017}]\label{dist}
Let $U$ and $V$ be two distinct elements in $\mathcal{G}_{q}(n,k)$. Then the following are equivalent.

$(1)$ $\dim(U\cap \alpha V)\leq 1$, for any $\alpha \in \mathbb{F}_{q^{n}}^{*}$.

$(2)$ For any nonzero elements $a$, $c\in U$ and nonzero elements $b$, $d\in V$, the equality $ab=cd$ implies that $a\mathbb{F}_{q}=c\mathbb{F}_{q}$ and $b\mathbb{F}_{q}=d\mathbb{F}_{q}$.
\end{Lemma}

Niu et al. \cite{NYW2020} gave several Sidon spaces with new parameters and provided some constructions of cyclic CDCs. Feng and Wang \cite{FW2021} provided a construction of cyclic $(n,2k-2,k)_q$-CDCs with size $(\lceil \frac{n}{2k} \rceil-1)\frac{q^{k}(q^{n}-1)}{q-1}$, where $n$ is a multiple of $k$ with $n\geq 3k$. Their construction  uses a variant of Sidon spaces in \cite{RRT2017}, and the size differs from the sphere-packing bound by a factor of $\frac{1}{q-1}$ as $k$ goes to infinity. Niu et al. \cite{NXG2022} generalized the constructions of \cite{FW2021}, and provided a construction of cyclic $(n,2k-2,k)_q$-CDCs with size $\ell(\lceil \frac{n}{4k} \rceil-1)\frac{q^{k}(q^{n}-1)}{q-1}$, where $n$ is a multiple of $k$ with $n\geq 5k$,  $\ell \leq k$ is the size of $S:=\{t_{1},t_{2},\ldots,t_{\ell}\}\subseteq \mathbb{N}$ such that $\gcd(t_{r},k)=1$ for all $1\leq r\leq \ell$ and $\gcd(t_{r_{2}}-t_{r_{1}},k)=1$ for all $t_{r_{1}}<t_{r_{2}}\in S$. Recently, Zhang and Tang \cite{ZT2023} generalized the construction in \cite{FW2021}  and  provided a construction of cyclic $(n,2k-2,k)_q$-CDCs with size $(\lceil \frac{n}{2k} \rceil-1)q^{k}(q^{n}-1)$, where $n$ is a multiple of $k$ with $n\geq 3k$.

In this paper, we present two constructions of Sidon spaces. The first one  makes use of  a variant of Sidon spaces constructed in \cite{FW2021} and \cite{ZT2023}, cf. Lemma \ref{NSidon}, while the second uses irreducible polynomials and a variant of Sidon spaces in \cite{FW2021}. By taking the union of optimal full-length orbit codes arising from Sidon spaces, two new classes of cyclic CDCs are obtained, i.e., there is a cyclic $(n,2k-2,k)_q$-CDC of size  $r(\lceil \frac{n}{2rk} \rceil-1)((q^{k}-1)^{r}(q^{n}-1)
+\frac{(q^{k}-1)^{r-1}(q^{n}-1)}{q-1})$ where $k|n$, $n\geq (2r+1)k$ and $r\geq 2$, and a cyclic $(n,2k-2,k)_q$-CDC of of size $\lfloor\frac{(r-1)(q^{k}-2)(q^{k}-1)^{r-1}(q^{n}-1)}{2}\rfloor$ where $n=2rk$ and $r\geq 2$. Consequently, several new lower bounds of cyclic CDCs are obtained, cf. Table I. In particular,  in the case of $n=4k$, the size of our resulting code is within a factor of $\frac{1}{2}+o_{k}(1)$ of the sphere-packing bound as $k$ goes to infinity.

\begin{figure}[htp]\label{Table}
 \begin{center}
 TABLE I: The lower bounds of cyclic subspace codes in $\mathcal{G}_{q}(n,k)$
 \end{center}
 \vskip 2.5mm
{\footnotesize
\begin{center}
  \begin{tabular}
{|c|c|c|c|c|c|}
		\hline
		$n$ & old lower bounds  & Our lower bounds & source\\ \hline
      $n=4k$& $q^{k}\cdot \frac{q^{n}-1}{q-1}$ \cite{FW2021} &$\lfloor\frac{(q^{k}-2)(q^{k}-1)(q^{n}-1)}{2}\rfloor$ & Theorem \ref{Con2} \\
      \hline
%	$n=(2r+1)k,r \geq 2$& $2\frac{q^{n}-1}{q-1}$ \cite{NYW2020} &$r(\frac{((q^{k}-1)(q-1)+1)(q^{k}-1)^{r-1}(q^{n}-1)}{q-1})$ \\	
%	\hline 
       $n=(2r+1)k,r \geq 2$& $rq^{k}(q^{n}-1)$ \cite{ZT2023} &$r(\frac{((q^{k}-1)(q-1)+1)(q^{k}-1)^{r-1}(q^{n}-1)}{q-1})$ & Theorem \ref{Con1} \\	
	\hline
%	$n=(2r+1)k,r \geq 2$ & $rq^{k}\cdot \frac{q^{n}-1}{q-1}$\cite{FW2021}&$r(\frac{((q^{k}-1)(q-1)+1)(q^{k}-1)^{r-1}(q^{n}-1)}{q-1})$\\
%	\hline 
%     $n=(2r+1)k,r \geq 2$& $\lfloor \frac{q-1}{2} \rfloor \frac{q^{n}-1}{q-1}$\cite{HC2021}& $r(\frac{((q^{k}-1)(q-1)+1)(q^{k}-1)^{r-1}(q^{n}-1)}{q-1})$\\
%	\hline
     $n=(2r+1)k,r \geq 2$& $\ell r\frac{q^{k}(q^{n}-1)}{q-1}$, $(\ell<k)$ \cite{NXG2022}&$r(\frac{((q^{k}-1)(q-1)+1)(q^{k}-1)^{r-1}(q^{n}-1)}{q-1})$ & Theorem \ref{Con1}\\
	\hline
\end{tabular}
\end{center}
}
\end{figure}

The remainder of this paper is organized as follows. In Section $2$, we present the first construction for cyclic  $(n,2k-2,k)_q$-CDCs where $k|n$, $n\geq (2r+1)k$ and $r\geq 2$. Section $3$ presents the second construction for cyclic  $(n,2k-2,k)_q$-CDCs where $n=2rk$ and $r\geq 2$. In Section $4$, we conclude this paper.

\section{The first construction}

Let $k,n$ be positive integers such that $n$ be a multiple of $k$ and $n\geq (2r+1)k$, where $r\geq 2$ is a positive integer. Let $\xi$ be a  primitive element in $\mathbb{F}_{q^{k}}$ and $\gamma$ a root of an irreducible polynomial of degree $\frac{n}{k}$ over $\mathbb{F}_{q^{k}}$, which lies in $\mathbb{F}_{q^{n}}$. Clearly $\{1,\gamma,\ldots, \gamma^{\frac{n}{k}-1}\}$ is a linearly independent set over $\mathbb{F}_{q^{k}}$. Let $i_{m}$ $(1\leq m \leq r)$, $b$, $\ell$ and $j$ be integers such that $0\leq i_{m}\leq q^{k}-2$ for $1\leq m \leq r$, $0\leq b \leq q-2$, $1\leq \ell \leq r$ and $1\leq j\leq \lceil \frac{n}{2rk}\rceil-1$. Define subspaces of dimension $k$ as follows:
\begin{equation}\label{Equ1}
U^{\ell}_{i_{1},i_{2},\ldots,i_{r},b,j}:=\{u+(u^{q}-\xi^{b}u)\xi^{i_{\ell}}\gamma^{\ell j}+\sum\limits_{m=1,m\neq \ell}^ru\xi^{i_{m}}\gamma^{mj}\colon u\in \mathbb{F}_{q^{k}}\},
\end{equation}
\begin{equation}\label{Equ2}
V^{\ell}_{i_{1},i_{2},\ldots,i_{r},j}:=\{v+v^{q}\gamma^{\ell j}+\sum\limits_{m=1,m\neq \ell}^rv\xi^{i_{m}}\gamma^{mj}\colon v\in \mathbb{F}_{q^{k}}\}.
\end{equation}

\begin{Construction}\label{Construction}
Let $k$ be a positive integer with $k\geq 2$, $n$ a multiple of $k$ such that $n\geq (2r+1)k$, where $r\geq 2$. Set  $e=\lceil \frac{n}{2rk}\rceil-1$.  For $0\leq i_{m}\leq q^{k}-2$ $(1\leq m \leq r)$, $0\leq b\leq q-2$, $1\leq j\leq e$ and $1\leq \ell \leq r$, define 
\[
\begin{array}{l}
\mathcal{C}^{\ell}_{i_{1},i_{2},\ldots,i_{r},b,j}:= \{\alpha U^{\ell}_{i_{1},i_{2},\ldots,i_{r},b,j}\colon \alpha \in \mathbb{F}_{q^{n}}^{*}\}, \smallskip \\
\mathcal{D}^{\ell}_{i_{1},i_{2},\ldots,i_{r},j}:= \{\alpha V^{\ell}_{i_{1},i_{2},\ldots,i_{r},j}\colon \alpha \in \mathbb{F}_{q^{n}}^{*}\},
\end{array}
\] where $U^{\ell}_{i_{1},i_{2},\ldots,i_{r},b,j}$ and $V^{\ell}_{i_{1},i_{2},\ldots,i_{r},j}$ are defined as in $(\ref{Equ1})$ and $(\ref{Equ2})$. Define
\[
\begin{array}{l}
\mathcal{C}^{\ell}:=\bigcup\limits_{j=1}^{e}\bigcup\limits_{0\leq i_{1},\ldots,i_r\leq q^k-2}\bigcup\limits_{b=0}^{q-2}\mathcal{C}^{\ell}_{i_{1},i_{2},\ldots,i_{r},b,j}, \smallskip \\
\mathcal{D}^{\ell}:=\bigcup\limits_{j=1}^{e}\bigcup\limits_{0\leq i_{1},\ldots,i_r\leq q^k-2}\mathcal{D}^{\ell}_{i_{1},i_{2},\ldots,i_{r},j},\smallskip\\
\mathcal{C}=\bigcup_{\ell=1}^{r}\mathcal{C}^{\ell}\ {\rm and}\ \mathcal{D}=\bigcup_{\ell=1}^{r}\mathcal{D}^{\ell}.
\end{array}
\] 

\end{Construction}

\begin{Theorem}\label{Con1}
Let $\mathcal{C}$ and $\mathcal{D}$ be defined as in Construction $\ref{Construction}$. Then $\mathcal{C}\cup \mathcal{D}$ is a cyclic $(n,2k-2,k)_q$-CDC of size $r(\lceil \frac{n}{2rk} \rceil-1)\left [(q^{k}-1)^{r}(q^{n}-1)+\frac{(q^{k}-1)^{r-1}(q^{n}-1)}{q-1}\right ]$.
\end{Theorem}

In order to prove Theorem \ref{Con1}, we give some  preliminary lemmas.

\begin{Lemma}\label{NSidon}
The subspaces defined in $(\ref{Equ1})$ and $(\ref{Equ2})$ are Sidon spaces.
\end{Lemma}

\begin{proof} We only prove that $U^{r}_{i_{1},i_{2},\ldots,i_{r},b,j}$ is a Sidon space. The other cases can be proved similarly.
Let
\[
\begin{array}{l}
\bar{u}_{1}=u_{1}+u_{1}\xi^{i_{1}}\gamma^{j}+\cdots+u_{1}\xi^{i_{r-1}}\gamma^{(r-1)j}+(u_{1}^{q}-\xi^{b}u_{1})\xi^{i_{r}}\gamma^{rj},\smallskip \\
\bar{v}_{1}=v_{1}+v_{1}\xi^{i_{1}}\gamma^{j}+\cdots+v_{1}\xi^{i_{r-1}}\gamma^{(r-1)j}+(v_{1}^{q}-\xi^{b}v_{1})\xi^{i_{r}}\gamma^{rj},\smallskip \\
\bar{u}_{2}=u_{2}+u_{2}\xi^{i_{1}}\gamma^{j}+\cdots+u_{2}\xi^{i_{r-1}}\gamma^{(r-1)j}+(u_{2}^{q}-\xi^{b}u_{2})\xi^{i_{r}}\gamma^{rj},\smallskip \\
\bar{v}_{2}=v_{2}+v_{2}\xi^{i_{1}}\gamma^{j}+\cdots+v_{2}\xi^{i_{r-1}}\gamma^{(r-1)j}+(v_{2}^{q}-\xi^{b}v_{2})\xi^{i_{r}}\gamma^{rj}
\end{array}
\]  
be four distinct nonzero elements of $U^{r}_{i_{1},i_{2},\ldots,i_{r},b,j}$, where $u_{1}$, $u_{2}$, $v_{1}$, $v_{2}$ are nonzero elements of $\mathbb{F}_{q^{k}}$. Suppose that $\bar{u}_{1}\bar{v}_{1}=\bar{u}_{2}\bar{v}_{2}$. We need to show that $\{\bar{u}_{1}\mathbb{F}_{q},\bar{v}_{1}\mathbb{F}_{q}\}=\{\bar{u}_{2}\mathbb{F}_{q},\bar{v}_{2}\mathbb{F}_{q}\}.$

Since $j\leq \lceil \frac{n}{2rk}\rceil-1$, we have $2rj\leq \frac{n}{k}-1$, thereby $1$, $\gamma^{j}$, $\gamma^{2j}$, $\ldots$, $\gamma^{2rj}$ are linearly independent over $\mathbb{F}_{q^{k}}$. Comparing the coefficients of $1,\gamma^{2rj}$ respectively in $\bar{u}_{1}\bar{v}_{1}=\bar{u}_{2}\bar{v}_{2}$ after expansion, we deduce that 
\begin{equation}\label{E1}
	u_{1}v_{1}=u_{2}v_{2},
\end{equation}
\begin{equation}\label{E2}
      u_{1}^{q}v_{1}+v_{1}^{q}u_{1}=u_{2}^{q}v_{2}+v_{2}^{q}u_{2}.
\end{equation}
Set $\frac{u_{1}}{u_{2}}=\lambda$. If $\lambda\in \mathbb{F}_{q}$, then $\frac{\bar{u}_{1}}{\bar{u}_{2}}, \frac{\bar{v}_{2}}{\bar{v}_{1}}\in \mathbb{F}_{q}$, thereby $\bar{u}_{1}\mathbb{F}_{q}=\bar{u}_{2}\mathbb{F}_{q}$ and $\bar{v}_{1}\mathbb{F}_{q}=\bar{v}_{2}\mathbb{F}_{q}$. If $\lambda \notin \mathbb{F}_{q}$, then after replacing $u_1,v_2$ by $\lambda u_{2}, \lambda v_{1}$ respectively, Equation (\ref{E2}) becomes the following: $$\lambda^{q}u_{2}^{q}v_{1}+\lambda v_{1}^{q}u_{2}=\lambda u_{2}^{q}v_{1}+\lambda^{q}v_{1}^{q}u_{2}.$$
Since $\lambda^{q}\neq \lambda$, we have $u_{2}^{q}v_{1}=u_{2}v_{1}^{q}$,  which
 implies that $\frac{u_{2}}{v_{1}}\in \mathbb{F}_{q}$,  $\bar{u}_{2}\mathbb{F}_{q}=\bar{v}_{1}\mathbb{F}_{q}$ and $\bar{u}_{1}\mathbb{F}_{q}=\bar{v}_{2}\mathbb{F}_{q}$.

So, the subspace $U^{r}_{i_{1},i_{2},\dots,i_{r},b,j}$ is a Sidon space of dimension $k$ over $\mathbb{F}_{q}$.\qed
\end{proof}

%Similar to the proof of Lemma \ref{NSidon}, the subspace in the following lemma is also a Sidon space. We omit the details. 

\begin{Lemma}\label{Main}
Let $\mathcal{C}^{\ell}$ be defined as in Construction $\ref{Construction}$. Then $\mathcal{C}^{\ell}$ is a cyclic $(n,2k-2,k)_q$-CDC of size $(\lceil \frac{n}{2rk} \rceil-1)(q^{k}-1)^{r}(q^{n}-1)$.
\end{Lemma}

\begin{proof}
By Lemma \ref{NSidon}, each $U^{\ell}_{i_{1},i_{2},\ldots,i_{r},b,j}$ is a Sidon space, then by Lemma \ref{sac} each $\mathcal{C}^{\ell}_{i_{1},i_{2},\ldots,i_{r},b,j}$ is an optimal full-length orbit code, i.e., a cyclic $(n,2k-2,k)_q$-CDC of size $\frac{q^{n}-1}{q-1}$. To show that $\mathcal{C}^{\ell}$ has minimum distance $2k-2$, it remains to show that $\dim(\alpha U^{\ell}_{i_{1},i_{2},\ldots,i_{r},b,j}\cap U^{\ell}_{i_{1}',i_{2}',\ldots,i_{r}',b',j'})\leq 1$ for any $(i_{1},i_{2},\ldots,i_{r},b,j)\neq (i_{1}',i_{2}',\ldots,i_{r}',b',j')$ and  $\alpha \in \mathbb{F}_{q^{n}}^{*}$. 
By Lemma \ref{dist}, it suffices to show that for any nonzero elements $\bar{u}_{1},\bar{u}_{2}\in U^{\ell}_{i_{1},i_{2},\ldots,i_{r},b,j}$  and nonzero elements $\bar{v}_{1},\bar{v}_{2}\in U^{\ell}_{i_{1}',i_{2}',\ldots,i_{r}',b',j'}$, the equality $\bar{u}_{1}\cdot \bar{v}_{1}=\bar{u}_{2}\cdot \bar{v}_{2}$ implies $\bar{u}_{1}\mathbb{F}_{q}=\bar{u}_{2}\mathbb{F}_{q}$ and $\bar{v}_{1}\mathbb{F}_{q}=\bar{v}_{2}\mathbb{F}_{q}$. 

Let $\bar{u}_{1}, \bar{u}_{2}, \bar{v}_{1}$ and $\bar{v}_{2}$ be of the following form:
\[
\begin{array}{l}
\bar{u}_{1}=u_{1}+u_{1}\xi^{i_{1}}\gamma^{j}+\cdots+(u_{1}^{q}-\xi^{b}u_{1})\xi^{i_{\ell}}\gamma^{\ell j}+\cdots+u_{1}\xi^{i_{r-1}}\gamma^{(r-1)j}+u_{1}\xi^{i_{r}}\gamma^{rj}, \smallskip \\		
\bar{u}_{2}=u_{2}+u_{2}\xi^{i_{1}}\gamma^{j}+\cdots+(u_{2}^{q}-\xi^{b}u_{2})\xi^{i_{\ell}}\gamma^{\ell j}+\cdots +u_{2}\xi^{i_{r-1}}\gamma^{(r-1)j}+u_{2}\xi^{i_{r}}\gamma^{rj}, \smallskip \\	
\bar{v}_{1}=v_{1}+v_{1}\xi^{i_{1}'}\gamma^{j'}+\cdots+(v_{1}^{q}-\xi^{b'}v_{1})\xi^{i_{\ell}'}\gamma^{\ell j'}+\cdots+v_{1}\xi^{i_{r-1}'}\gamma^{(r-1)j'}+v_{1}\xi^{i_{r}'}\gamma^{rj'}, \smallskip \\	
\bar{v}_{2}=v_{2}+v_{2}\xi^{i_{1}'}\gamma^{j'}+\cdots+(v_{2}^{q}-\xi^{b'}v_{2})\xi^{i_{\ell}'}\gamma^{\ell j'}+\cdots+v_{2}\xi^{i_{r-1}'}\gamma^{(r-1)j'}+v_{2}\xi^{i_{r}'}\gamma^{rj'},\end{array}
\] 
where $u_{1}, u_{2}, v_{1}$ and $v_{2}$ are nonzero elements of $\mathbb{F}_{q^{k}}$. Simple computation shows that 
\begin{equation*}
\begin{split}
&\quad \ \bar{u}_{1}\cdot \bar{v}_{1}\\
&=u_1v_1+u_1(v_1^q-v_1\xi^{b'})\xi^{i_{\ell}'}\gamma^{\ell j'}+\sum_{t=1,t\neq \ell}^{r}u_1v_1\xi^{i_{t}'}\gamma^{tj'}+(u_1^q-\xi^{b}u_1)v_1\xi^{i_{\ell}}\gamma^{\ell j}\\
&\quad + (u_1^q-\xi^{b}u_1)(v_1^q-\xi^{b'}v_1)\xi^{i_{\ell}+i_{\ell}'}\gamma^{\ell j+\ell j'} +\sum_{t=1,t\neq \ell}^{r}(u_1^q-\xi^{b}u_1)v_1\xi^{i_{\ell}+i_{t}'}\gamma^{\ell j+tj'}\\
&\quad +\sum_{m=1,m\neq \ell}^{r}u_1v_1\xi^{i_{m}}\gamma^{mj}+\sum_{m=1,m\neq \ell}^{r}u_1(v_1^q-\xi^{b'}v_1)\xi^{i_{m}+i_{\ell}'}\gamma^{mj+\ell j'}\\
&\quad+\sum_{m=1,m\neq \ell}^r\sum_{t=1,t\neq \ell}^{r}u_1v_1\xi^{i_{m}+i_{t}'}\gamma^{mj+tj'}.
\end{split}
\end{equation*}
Since $1,\gamma,\ldots,\gamma^{2r\max\{j,j'\}}$ are linearly independent over $\mathbb{F}_{q^k}$, comparing coefficients of $\bar{u}_{1}\cdot \bar{v}_{1}$ and $\bar{u}_{2}\cdot \bar{v}_{2}$ after expansion gives
\begin{equation}\label{b1}
	u_{1}v_{1}=u_{2}v_{2},
\end{equation} 
and the equality $\bar{u}_{1}\cdot \bar{v}_{1}=\bar{u}_{2}\cdot \bar{v}_{2}$ can be simplified as follows:
\begin{equation}\label{b0}
\begin{split}
&\quad u_1v_1^q\xi^{i_{\ell}'}\gamma^{\ell j'}+u_1^qv_1\xi^{i_{\ell}}\gamma^{\ell j}-(u_1^qv_1\xi^{b'}+u_1v_1^q\xi^{b})\xi^{i_{\ell}+i_{\ell}'}\gamma^{\ell j+\ell j'}\\ &\quad +\sum_{t=1,t\neq \ell}^{r}u_1^qv_1\xi^{i_{\ell}+i_{t}'}\gamma^{\ell j+tj'}+\sum_{m=1,m\neq \ell}^{r}u_1v_1^q\xi^{i_{m}+i_{\ell}'}\gamma^{mj+\ell j'}\\
&= u_2v_2^q\xi^{i_{\ell}'}\gamma^{\ell j'}+u_2^qv_2\xi^{i_{\ell}}\gamma^{\ell j}-(u_2^qv_2\xi^{b'}+u_2v_2^q\xi^{b})\xi^{i_{\ell}+i_{\ell}'}\gamma^{\ell j+\ell j'}\\ &\quad +\sum_{t=1,t\neq \ell}^{r}u_2^qv_2\xi^{i_{\ell}+i_{t}'}\gamma^{\ell j+tj'}+\sum_{m=1,m\neq \ell}^{r}u_2v_2^q\xi^{i_{m}+i_{\ell}'}\gamma^{mj+\ell j'}.
\end{split}
\end{equation}
%Simple computation shows that 
%\begin{equation*}
%\begin{split}
%\bar{u}_{1}\cdot \bar{v}_{1}=& u_{1}v_{1}+u_{1}v_{1}\xi^{i_{1}'}\gamma^{j'}+\cdots+u_{1}(v_{1}^{q}-v_{1})\xi^{i_{\ell}'}\gamma^{\ell j'}\\
%&+\cdots+u_{1}v_{1}\xi^{i_{r}'}\gamma^{rj'}+u_{1}v_{1}\xi^{i_{1}}\gamma^{j}+u_{1}v_{1}\xi^{i_{1}+i_{1}'}\gamma^{j+j'}\\
%&+\cdots+u_{1}(v_{1}^{q}-v_{1})\xi^{i_{1}+i_{\ell}'}\gamma^{j+\ell j'}+\cdots+u_{1}v_{1}\xi^{i_{1}+i_{r}'}\gamma^{j+rj'}\\
%&+\cdots+(u_{1}^{q}-u_{1})v_{1}\xi^{i_{\ell}}\gamma^{\ell j}	+(u_{1}^{q}-u_{1})v_{1}\xi^{i_{1}'+i_{\ell}}\gamma^{\ell j+j'}\\
%&+\cdots+(u_{1}^{q}-u_{1})(v_{1}^{q}-v_{1})\xi^{i_{\ell}'	+i_{\ell}}\gamma^{\ell j+\ell' j'}+\cdots\\
%&+(u_{1}^{q}-u_{1})(v_{1}^{q}-v_{1})\xi^{i_{\ell}+i_{r}'}\gamma^{\ell j+rj'}+\cdots+u_{1}v_{1}\xi^{i_{r}}\gamma^{rj}\\
%&+u_{1}v_{1}\xi^{i_{1}'+i_{r}}\gamma^{j'+rj}+\cdots+u_{1}(v_{1}^{q}-v_{1})\xi^{i_{\ell}'+i_{r}}\gamma^{\ell j'+rj}+\cdots+u_{1}v_{1}\xi^{i_{r}'+i_{r}}\gamma^{rj'+rj}.
%\end{split}
%\end{equation*}
We distinguish the following two cases.

Case 1:  $j\neq j'$.

Without loss of generality, we can assume that $j<j'$. Notice that $\ell j< \min \{\ell j+xj'\colon 1\leq x\leq r\}$ and  $\ell j<\min \{\ell j'+yj\colon 0\leq y\leq r\}$. Comparing coefficients of $\gamma^{\ell j}$ in Equality (\ref{b0}) gives $u_1v_1=u_2v_2$ and $u_{1}^{q}v_{1}=u_{2}^{q}v_{2}$.
Then $\frac{v_{1}}{v_{2}}=(\frac{u_{2}}{u_{1}})^{q}=(\frac{v_{1}}{v_{2}})^{q}$, thereby $\frac{v_{1}}{v_{2}}\in \mathbb{F}_{q}^{*}$. It follows that $\bar{u}_{1}\mathbb{F}_{q}=\bar{u}_{2}\mathbb{F}_{q}$ and $\bar{v}_{1}\mathbb{F}_{q}=\bar{v}_{2}\mathbb{F}_{q}$, as desired.

Case 2:  $j=j'$. 

Since $(i_{1},i_{2},\ldots,i_{r},b,j)\neq (i_{1}',i_{2}',\ldots,i_{r}',b',j')$, we have that  $$(i_{1},i_{2},\ldots, i_{r},b)\neq(i_{1}',i_{2}',\ldots,i_{r}',b').$$
Since $2re\leq \frac{n}{k}-1$, we have that $1$, $\gamma^{j}$, $\gamma^{2j}$, $\ldots$, $\gamma^{2rj}$ are linearly independent over $\mathbb{F}_{q^{k}}$. Comparing coefficients in Equality (\ref{b0})  gives
\begin{equation}\label{b2}
		u_{1}v_{1}^{q}\xi^{b'}+u_{1}^{q}v_{1}\xi^{b}=u_{2}v_{2}^{q}\xi^{b'}+u_{2}^{q}v_{2}\xi^{b},
\end{equation}
\begin{equation}\label{b4}
		u_{1}v_{1}^{q}\xi^{i_{\ell}'}+u_{1}^{q}v_{1}\xi^{i_{\ell}}=u_{2}v_{2}^{q}\xi^{i_{\ell}'}+u_{2}^{q}v_{2}\xi^{i_{\ell}},
\end{equation}
\begin{equation}\label{b3}
		u_{1}v_{1}^{q}\xi^{i_{x}+i_{\ell}'}+u_{1}^{q}v_{1}\xi^{i_{x}'+i_{\ell}}=u_{2}v_{2}^{q}\xi^{i_{x}+i_{\ell}'}+u_{2}^{q}v_{2}\xi^{i_{x}'+i_{\ell}}
\end{equation}
for $1\leq x\leq r$ and $x\neq \ell$ .

If $b\neq b'$, we set $\frac{u_{1}}{u_{2}}=\lambda\in \mathbb{F}_{q^{k}}$. If $\lambda\notin \mathbb{F}_{q}$, by Equation (\ref{b2}) along with $u_{1}=\lambda u_{2}$ and $v_{2}=\lambda v_{1}$, we can duduce that $\xi^{b'-b}=(\frac{u_{2}}{v_{1}})^{q-1}$, a contradiction to $0\leq b\neq b'\leq q-2$. Hence, $\lambda\in \mathbb{F}_{q}$. Then $\bar{u}_{1}\mathbb{F}_{q}=\bar{u}_{2}\mathbb{F}_{q}$ and $\bar{v}_{1}\mathbb{F}_{q}=\bar{v}_{2}\mathbb{F}_{q}$, as desired.

If $b=b'$ and $i_{\ell}\neq i_{\ell}'$, Equations (\ref{b2}) and (\ref{b4}) imply $u_{1}v_{1}^{q}(\xi^{i_{\ell}'-i_{\ell}}-1)=u_{2}v_{2}^{q}(\xi^{i_{\ell}'-i_{\ell}}-1)$. Since $\xi$ is a primitive element of $\mathbb{F}_{q^{k}}$, $0\leq i_{\ell},i_{\ell}' \leq q^{k}-2$ and $i_{\ell}\neq i_{\ell}'$, we deduce that $\xi^{i_{\ell}'-i_{\ell}}\neq 1$, thereby $u_{1}v_{1}^{q}=u_{2}v_{2}^{q}$. It follows from Equation (\ref{b1}) that $\frac{u_{1}}{u_{2}}=(\frac{v_{2}}{v_{1}})^{q}=(\frac{u_{1}}{u_{2}})^{q}$, thereby $\frac{u_{1}}{u_{2}}\in \mathbb{F}_{q}^{*}$. We then deduce that $\bar{u}_{1}\mathbb{F}_{q}=\bar{u}_{2}\mathbb{F}_{q}$ and $\bar{v}_{1}\mathbb{F}_{q}=\bar{v}_{2}\mathbb{F}_{q}$, as desired.

If $b=b'$ and $i_{\ell}= i_{\ell}'$, then we have $i_{y}\neq i_{y}'$ for some $y$ with $1\leq y\leq r$ and $y\neq \ell$. By Equations (\ref{b4}) and  (\ref{b3}) with $x=y$, we can duduce that $u_{1}v_{1}^{q}(\xi^{i_{y}'-i_{y}}-1)=u_{2}v_{2}^{q}(\xi^{i_{y}'-i_{y}}-1)$. Since $\xi$ is a primitive element of $\mathbb{F}_{q^{k}}$, $0\leq i_{y}, i_{y}' \leq q^{k}-2$ and $i_{y}\neq i_{y}'$, we deduce that $\xi^{i_{y}'-i_{y}}\neq 1$ and $u_{1}v_{1}^{q}=u_{2}v_{2}^{q}$.   It follows from Equation (\ref{b1}) that $\frac{u_{1}}{u_{2}}=(\frac{v_{2}}{v_{1}})^{q}=(\frac{u_{1}}{u_{2}})^{q}$, thereby $\frac{u_{1}}{u_{2}}\in \mathbb{F}_{q}^{*}$. We then deduce that $\bar{u}_{1}\mathbb{F}_{q}=\bar{u}_{2}\mathbb{F}_{q}$ and $\bar{v}_{1}\mathbb{F}_{q}=\bar{v}_{2}\mathbb{F}_{q}$, as desired.

To sum up, we have shown that $\mathcal{C}^{\ell}$ has minimum distance $2k-2$. In particular, we have shown that the $(\lceil \frac{n}{2rk} \rceil-1)(q^{k}-1)^{r}(q^{n}-1)$ elements of $\mathcal{C}^{\ell}$ are distinct, so $|\mathcal{C}^{\ell}|=(\lceil \frac{n}{2rk} \rceil-1)(q^{k}-1)^{r}(q^{n}-1)$. \qed
\end{proof}

\begin{Lemma}\label{CSC-C}
$\mathcal{C}$ is a cyclic $(n,2k-2,k)_q$-CDC of size $r(\lceil \frac{n}{2rk} \rceil-1)(q^{k}-1)^{r}(q^{n}-1)$.
\end{Lemma}

\begin{proof}
Each $\mathcal{C}^{\ell}$ is a cyclic $(n,2k-2,k)_q$-CDC of size $(\lceil \frac{n}{2rk} \rceil-1)(q^{k}-1)^{r}(q^{n}-1)$ by Lemma \ref{Main}. To show that $\mathcal{C}$ has minimum distance $2k-2$, it remains to show that $\dim(\alpha U^{\ell}_{i_{1},i_{2},\ldots,i_{r},b,j}\cap U^{\ell'}_{i_{1}',i_{2}',\ldots,i_{r}',b',j'})\leq 1$ for any $\alpha \in \mathbb{F}_{q^{n}}^{*}$, $1\leq \ell< \ell'\leq r$, $0\leq i_{m},i_{m}'\leq q^{k}-2$ $(1\leq m \leq r)$, $0\leq b,b'\leq q-2$, and $1\leq j,j'\leq e$. By Lemma \ref{dist}, it suffices to show that for any nonzero elements $\bar{u}_{1},\bar{u}_{2}\in U^{\ell}_{i_{1},i_{2},\ldots,i_{r},b,j}$  and nonzero elements $\bar{v}_{1},\bar{v}_{2}\in U^{\ell'}_{i_{1}',i_{2}',\ldots,i_{r}',b',j'}$, the equality $\bar{u}_{1}\cdot \bar{v}_{1}=\bar{u}_{2}\cdot \bar{v}_{2}$ implies $\bar{u}_{1}\mathbb{F}_{q}=\bar{u}_{2}\mathbb{F}_{q}$, $\bar{v}_{1}\mathbb{F}_{q}=\bar{v}_{2}\mathbb{F}_{q}$. 

Let $\bar{u}_{1}, \bar{u}_{2}, \bar{v}_{1}$ and $\bar{v}_{2}$ be of the following form:
\[
\begin{array}{l}
\bar{u}_{1}=u_{1}+u_{1}\xi^{i_{1}}\gamma^{j}+\cdots+(u_{1}^{q}-\xi^{b}u_{1})\xi^{i_{\ell}}\gamma^{\ell j}+\cdots+u_{1}\xi^{i_{r}}\gamma^{rj}, \smallskip \\
\bar{u}_{2}=u_{2}+u_{2}\xi^{i_{1}}\gamma^{j}+\cdots+(u_{2}^{q}-\xi^{b}u_{2})\xi^{i_{\ell}}\gamma^{\ell j}+\cdots+u_{2}\xi^{i_{r}}\gamma^{rj}, \smallskip \\
\bar{v}_{1}=v_{1}+v_{1}\xi^{i_{1}'}\gamma^{j'}+\cdots+(v_{1}^{q}-\xi^{b'}v_{1})\xi^{i_{\ell'}'}\gamma^{\ell' j'}+\cdots+v_{1}\xi^{i_{r}'}\gamma^{rj'}, \smallskip \\
\bar{v}_{2}=v_{2}+v_{2}\xi^{i_{2}'}\gamma^{j'}+\cdots+(v_{2}^{q}-\xi^{b'}v_{2})\xi^{i_{\ell'}'}\gamma^{\ell' j'}+\cdots+v_{2}\xi^{i_{r}'}\gamma^{rj'},
\end{array}
\] 
where $u_{1}$, $u_{2}$, $v_{1}$, $v_{2}$ are nonzero elements of $\mathbb{F}_{q^{k}}$.   Simple computation shows that 
\begin{equation*}
	\begin{split}
		&\quad \ \bar{u}_{1}\cdot \bar{v}_{1}\\
		&=u_1v_1+u_1(v_1^q-\xi^{b'}v_1)\xi^{i_{\ell'}'}\gamma^{\ell' j'}+\sum_{t=1,t\neq \ell'}^{r}u_1v_1\xi^{i_{t}'}\gamma^{tj'}+(u_1^q-\xi^{b}u_1)v_1\xi^{i_{\ell}}\gamma^{\ell j}\\
		&\quad + (u_1^q-\xi^{b}u_1)(v_1^q-\xi^{b'}v_1)\xi^{i_{\ell}+i_{\ell'}'}\gamma^{\ell j+\ell' j'} +\sum_{t=1,t\neq \ell'}^{r}(u_1^q-\xi^{b}u_1)v_1\xi^{i_{\ell}+i_{t}'}\gamma^{\ell j+tj'}\\
		&\quad +\sum_{m=1,m\neq \ell}^{r}u_1v_1\xi^{i_{m}}\gamma^{mj}+\sum_{m=1,m\neq \ell}^{r}u_1(v_1^q-\xi^{b'}v_1)\xi^{i_{m}+i_{\ell'}'}\gamma^{mj+\ell' j'}\\
            &\quad +\sum_{m=1,m\neq \ell}^r\sum_{t=1,t\neq \ell'}^{r}u_1v_1\xi^{i_{m}+i_{t}'}\gamma^{mj+tj'}.
	\end{split}
\end{equation*}
Since $1,\gamma,\ldots,\gamma^{2r\max\{j,j'\}}$ are linearly independent over $\mathbb{F}_{q^k}$, comparing coefficients of $\bar{u}_{1}\cdot \bar{v}_{1}$ and $\bar{u}_{2}\cdot \bar{v}_{2}$ after expansion gives $u_1v_1=u_2v_2$, and the equality $\bar{u}_{1}\cdot \bar{v}_{1}=\bar{u}_{2}\cdot \bar{v}_{2}$ can be simplified as follows:
\begin{equation}\label{b6}
\begin{split}
&\quad u_1v_1^q\xi^{i_{\ell'}'}\gamma^{\ell' j'}+u_1^qv_1\xi^{i_{\ell}}\gamma^{\ell j}-(u_1^qv_1\xi^{b'}+u_1v_1^q\xi^{b})\xi^{i_{\ell}+i_{\ell'}'}\gamma^{\ell j+\ell' j'}\\ &\quad +\sum_{t=1,t\neq \ell'}^{r}u_1^qv_1\xi^{i_{\ell}+i_{t}'}\gamma^{\ell j+tj'}+\sum_{m=1,m\neq \ell}^{r}u_1v_1^q\xi^{i_{m}+i_{\ell'}'}\gamma^{mj+\ell' j'}\\
&= u_2v_2^q\xi^{i_{\ell'}'}\gamma^{\ell' j'}+u_2^qv_2\xi^{i_{\ell}}\gamma^{\ell j}-(u_2^qv_2\xi^{b'}+u_2v_2^q\xi^{b})\xi^{i_{\ell}+i_{\ell'}'}\gamma^{\ell j+\ell' j'}\\ &\quad +\sum_{t=1,t\neq \ell'}^{r}u_2^qv_2\xi^{i_{\ell}+i_{t}'}\gamma^{\ell j+tj'}+\sum_{m=1,m\neq \ell}^{r}u_2v_2^q\xi^{i_{m}+i_{\ell'}'}\gamma^{mj+\ell' j'}.
	\end{split}
\end{equation}
We distinguish the following two cases

%Simple comptation shows that 
%\begin{equation*}
%\begin{split}
%\bar{u}_{1}\cdot \bar{v}_{1}=& u_{1}v_{1}+u_{1}v_{1}\xi^{i_{1}'}\gamma^{j'}+\cdots+u_{1}(v_{1}^{q}-v_{1})\xi^{i_{\ell'}'}\gamma^{\ell' j'}\\
%&+\cdots+u_{1}v_{1}\xi^{i_{r}'}\gamma^{rj'}+u_{1}v_{1}\xi^{i_{1}}\gamma^{j}+u_{1}v_{1}\xi^{i_{1}+i_{1}'}\gamma^{j+j'}\\
%&+\cdots+u_{1}(v_{1}^{q}-v_{1})\xi^{i_{1}+i_{\ell'}'}\gamma^{j+\ell' j'}+\cdots+u_{1}v_{1}\xi^{i_{1}+i_{r}'}\gamma^{j+rj'}\\
%&+\cdots+(u_{1}^{q}-u_{1})v_{1}\xi^{i_{\ell}}\gamma^{\ell j}	+(u_{1}^{q}-u_{1})v_{1}\xi^{i_{1}'+i_{\ell}}\gamma^{\ell j+j'}\\
%&+\cdots+(u_{1}^{q}-u_{1})(v_{1}^{q}-v_{1})\xi^{i_{\ell'}'	+i_{\ell}}\gamma^{\ell j+\ell' j'}+\cdots\\
%&+(u_{1}^{q}-u_{1})(v_{1}^{q}-v_{1})\xi^{i_{\ell}+i_{r}'}\gamma^{\ell j+rj'}+\cdots+u_{1}v_{1}\xi^{i_{r}}\gamma^{rj}\\
%&+u_{1}v_{1}\xi^{i_{1}'+i_{r}}\gamma^{j'+rj}+\cdots+u_{1}(v_{1}^{q}-v_{1})\xi^{i_{\ell'}'+i_{r}}\gamma^{\ell' j'+rj}+\cdots+u_{1}v_{1}\xi^{i_{r}'+i_{r}}\gamma^{rj'+rj}.
%\end{split}
%\end{equation*}

Case 1: $j\leq j'$.

Since $\ell j< \min \{\ell j+xj'\colon 1\leq x\leq r\}$ and  $\ell j<\min \{\ell' j'+yj\colon 0\leq y\leq r\}$, comparing coefficients of $\gamma^{\ell j}$ in Equation (\ref{b6}) gives $u_{1}^{q}v_{1}=u_{2}^{q}v_{2}$.
Since $u_1v_1=u_2v_2$, we deduce that $\frac{v_{1}}{v_{2}}=(\frac{u_{2}}{u_{1}})^{q}=(\frac{v_{1}}{v_{2}})^{q}$, thereby $\frac{v_{1}}{v_{2}}\in \mathbb{F}_{q}^{*}$. It follows that $\bar{u}_{1}\mathbb{F}_{q}=\bar{u}_{2}\mathbb{F}_{q}$ and $\bar{v}_{1}\mathbb{F}_{q}=\bar{v}_{2}\mathbb{F}_{q}$, as desired.

Case 2: $j> j'$. 

 If $\ell j\notin  \{\ell'j'+yj\colon 0\leq y\leq r\}$, comparing coefficients of $\gamma^{\ell j}$ in Equation (\ref{b6}) gives $u_{1}^{q}v_{1}=u_{2}^{q}v_{2}$.
 Since $u_1v_1=u_2v_2$, we deduce that $\frac{v_{1}}{v_{2}}=(\frac{u_{2}}{u_{1}})^{q}=(\frac{v_{1}}{v_{2}})^{q}$, thereby $\frac{v_{1}}{v_{2}}\in \mathbb{F}_{q}^{*}$. It follows that $\bar{u}_{1}\mathbb{F}_{q}=\bar{u}_{2}\mathbb{F}_{q}$ and $\bar{v}_{1}\mathbb{F}_{q}=\bar{v}_{2}\mathbb{F}_{q}$, as desired.

 If $\ell j=yj+\ell'j'$ for some $y$ with $0\leq y\leq \ell-1$, then $\ell'j'=(\ell-y)j$. We distinguish the following two subcases.

Subcase (i): $y\neq 0$.

Since $\ell'j'<\min\{\ell j+xj'\colon 0\leq x\leq r\}$ and $\ell' j'<\min \{\ell'j'+zj\colon 1\leq z\leq r\}$,  comparing coefficients of $\gamma^{\ell' j'}$ in Equation (\ref{b6}) gives $u_{1}v_{1}^{q}=u_{2}v_{2}^{q}$.
It follows from $u_{1}v_{1}=u_{2}v_{2}$ that $\frac{u_{1}}{u_{2}}=(\frac{v_{2}}{v_{1}})^{q}=(\frac{u_{1}}{u_{2}})^{q}$. Then $\frac{u_{2}}{u_{1}}\in \mathbb{F}_{q}^{*}$, which implies that $\bar{v}_{1}\mathbb{F}_{q}=\bar{v}_{2}\mathbb{F}_{q}$ and $\bar{u}_{1}\mathbb{F}_{q}=\bar{u}_{2}\mathbb{F}_{q}$, as desired.

Subcase (ii): $y=0$. 

Note that $y=0$ implies $\ell j=\ell'j'$. Since $\ell j+j'>\ell j=\ell'j'$,  $\ell j+j'< \min\{\ell j+xj'\colon 2\leq x\leq r\}$ and $\ell j+j'< \min \{\ell'j'+zj\colon 1\leq z\leq r\}$, comparing coefficient of $\gamma^{\ell j+j'}$ in Equation (\ref{b6}) gives $v_{1}u_{1}^{q}=v_{2}u_{2}^{q}$.
It follows from $u_{1}v_{1}=u_{2}v_{2}$ that $\frac{v_{1}}{v_{2}}=(\frac{u_{2}}{u_{1}})^{q}=(\frac{v_{1}}{v_{2}})^{q}$. Then $\frac{u_{2}}{u_{1}}\in \mathbb{F}_{q}^{*}$, which implies that $\bar{u}_{1}\mathbb{F}_{q}=\bar{u}_{2}\mathbb{F}_{q}$ and $\bar{v}_{1}\mathbb{F}_{q}=\bar{v}_{2}\mathbb{F}_{q}$, as desired.

To sum up, we have shown that $\mathcal{C}$ has minimum distance $2k-2$. In particular, we have shown that the $r(\lceil \frac{n}{2rk} \rceil-1)(q^{k}-1)^{r}(q^{n}-1)$ elements of $\mathcal{C}$ are distinct, so $|\mathcal{C}|=r(\lceil \frac{n}{2rk} \rceil-1)(q^{k}-1)^{r}(q^{n}-1)$. \qed
\end{proof}

Similar discussion to the proofs of Lemmas \ref{Main} and \ref{CSC-C}, we have cheaked that $\mathcal{D}^{\ell}$ and $\mathcal{D}$ have the same properties as $\mathcal{C}^{\ell}$ and $\mathcal{C
}$, respectively.
\begin{Lemma}\label{csc}
Let $\mathcal{D}^{\ell}$ be defined as in Construction $\ref{Construction}$. Then $\mathcal{D}^{\ell}$ is a cyclic $(n,2k-2,k)_q$-CDC of size $(\lceil \frac{n}{2rk} \rceil-1)\frac{(q^{k}-1)^{r-1}(q^{n}-1)}{q-1}$.
\end{Lemma}

\begin{Lemma}\label{CSC-D}
Let $\mathcal{D}$ be defined as in Construction $\ref{Construction}$. Then $\mathcal{D}$ is a cyclic $(n,2k-2,k)_q$-CDC  of size $r(\lceil \frac{n}{2rk} \rceil-1)\frac{(q^{k}-1)^{r-1}(q^{n}-1)}{q-1}$.
\end{Lemma}

We are in a position to prove Theorem \ref{Con1}.

\noindent {\bf{Proof of Theorem \ref{Con1}}:} By Lemmas \ref{CSC-C} and \ref{CSC-D}, 
to show that $\mathcal{C}\cup \mathcal{D}$ has minimum distance $2k-2$, it remains to show that $\dim(\alpha U^{\ell}_{i_{1},i_{2},\ldots,i_{r},b,j}\cap V^{\ell'}_{i_{1}',i_{2}',\ldots,i_{r}',j'})\leq 1$ for any $\alpha \in \mathbb{F}_{q^{n}}^{*}$, $0\leq i_m,i_m'\leq q^k-2$, $1\leq m\leq r$, $0\leq b\leq q-2$ and $1\leq j,j'\leq \lceil \frac{n}{2rk} \rceil-1$. By Lemma \ref{dist}, it suffices to show that for any nonzero elements $\bar{u}_{1},\bar{u}_{2}\in U^{\ell}_{i_{1},i_{2},\ldots,i_{r},b,j}$  and nonzero elements $\bar{v}_{1},\bar{v}_{2}\in V^{\ell'}_{i_{1}',i_{2}',\ldots,i_{r}',j'}$, the equality $\bar{u}_{1}\cdot \bar{v}_{1}=\bar{u}_{2}\cdot \bar{v}_{2}$ implies $\bar{u}_{1}\mathbb{F}_{q}=\bar{u}_{2}\mathbb{F}_{q}$, $\bar{v}_{1}\mathbb{F}_{q}=\bar{v}_{2}\mathbb{F}_{q}$. 
Let $\bar{u}_{1}, \bar{u}_{2}, \bar{v}_{1}$ and $\bar{v}_{2}$ be of the following form:
\[
\begin{array}{l}
\bar{u}_{1}=u_1+(u_1^{q}-\xi^{b}u_1)\xi^{i_{\ell}}\gamma^{\ell j}+\sum_{m=1,m\neq \ell}^ru_1\xi^{i_{m}}\gamma^{mj}, \smallskip \\		
\bar{u}_{2}=u_2+(u_2^{q}-\xi^{b}u_2)\xi^{i_{\ell}}\gamma^{\ell j}+\sum_{m=1,m\neq \ell}^ru_2\xi^{i_{m}}\gamma^{mj}, \smallskip \\
\bar{v}_{1}=v_{1}+v_{1}^{q}\gamma^{\ell' j'}+\sum_{m=1,m\neq \ell'}^rv_{1}\xi^{i_{m}'}\gamma^{mj'}, \smallskip \\	
\bar{v}_{2}=v_{2}+v_{2}^{q}\gamma^{\ell' j'}+\sum_{m=1,m\neq \ell'}^rv_{2}\xi^{i_{m}'}\gamma^{mj'},
\end{array}
\] 
where $u_{1}$, $u_{2}$, $v_{1}$, $v_{2}$ are nonzero elements of $\mathbb{F}_{q^{k}}$. 
Since $1,\gamma,\ldots,\gamma^{2r\max\{j,j'\}}$ are linearly independent over $\mathbb{F}_{q^k}$, comparing coefficients of $\bar{u}_{1}\cdot \bar{v}_{1}$ and $\bar{u}_{2}\cdot \bar{v}_{2}$ after expansion gives $u_1v_1=u_2v_2$, and the equality $\bar{u}_{1}\cdot \bar{v}_{1}=\bar{u}_{2}\cdot \bar{v}_{2}$ can be simplified as follows:
\begin{equation}\label{b10}
	\begin{split}
		&\quad u_1v_1^q\gamma^{\ell' j'}+u_1^qv_1\xi^{i_{\ell}}\gamma^{\ell j} -\xi^{b}u_1v_1^q\xi^{i_{\ell}}\gamma^{\ell j+\ell'j'}\\
		& \quad +\sum_{m=1,m\neq \ell'}^{r}u_1^qv_1\xi^{i_{\ell}+i_{m}'}\gamma^{\ell j+mj'}+\sum_{m=1,m\neq \ell}^r u_1v_1^q\xi^{i_{m}}\gamma^{mj+\ell' j'}\\
		&= u_2v_2^q\gamma^{\ell' j'}+u_2^qv_2\xi^{i_{\ell}}\gamma^{\ell j} -\xi^{b}u_2v_2^q\xi^{i_{\ell}}\gamma^{\ell j+\ell'j'}\\
		& \quad +\sum_{m=1,m\neq \ell'}^ru_2^qv_2\xi^{i_{\ell}+i_{m}'}\gamma^{\ell j+mj'}+\sum_{m=1,m\neq \ell}^r u_2v_2^q\xi^{i_{m}}\gamma^{mj+\ell' j'}\\
	\end{split}
\end{equation}  
%If $\ell j\neq \ell'j'$, we can assume that $\ell j<\ell' j'$.  Since 
%Notice that $\ell j<\ell j+xj'<$ and $\ell j\leq xj+\ell'j'$ for all $1\leq x\leq r$, comparing coefficients of $\gamma^{\ell j}$ in Equatity (\ref{b10}) gives $u_{1}^{q}v_{1}=u_{2}^{q}v_{2}$.
%Since $u_1v_1=u_2v_2$, we deduce that $\frac{v_{1}}{v_{2}}=(\frac{u_{2}}{u_{1}})^{q}=(\frac{v_{1}}{v_{2}})^{q}$, thereby $\frac{v_{1}}{v_{2}}\in \mathbb{F}_{q}^{*}$. It follows that $\bar{u}_{1}\mathbb{F}_{q}=\bar{u}_{2}\mathbb{F}_{q}$ and $\bar{v}_{1}\mathbb{F}_{q}=\bar{v}_{2}\mathbb{F}_{q}$ as desired.

%If $\ell j\neq \ell'j'$, 
Notice that $\ell j+\ell'j'\not \in \{\ell j+mj'\colon 0\leq m\leq r,m\neq \ell'\}\cup \{mj+\ell'j'\colon 0\leq m\leq r,m\neq \ell\}$. Comparing coefficients of $\gamma^{\ell j+\ell'j'}$ in Equatity (\ref{b10}) gives $u_{1}v_{1}^q=u_{2}v_{2}^q$.
Since $u_1v_1=u_2v_2$, we deduce that $\frac{v_{1}}{v_{2}}=\frac{u_{2}}{u_{1}}=(\frac{v_{1}}{v_{2}})^{q}$, thereby $\frac{v_{1}}{v_{2}}\in \mathbb{F}_{q}^{*}$. It follows that $\bar{u}_{1}\mathbb{F}_{q}=\bar{u}_{2}\mathbb{F}_{q}$ and $\bar{v}_{1}\mathbb{F}_{q}=\bar{v}_{2}\mathbb{F}_{q}$, as desired.

To sum up, we have shown that  $\mathcal{C}\cup \mathcal{D}$ has minimum distance $2k-2$. In particular, we have shown that the $r(\lceil \frac{n}{2rk} \rceil-1)((q^{k}-1)^{r}(q^{n}-1)+\frac{(q^{k}-1)^{r-1}(q^{n}-1)}{q-1})$ elements of  $\mathcal{C}\cup \mathcal{D}$ are distinct, so $| \mathcal{C}\cup \mathcal{D}|=r(\lceil \frac{n}{2rk} \rceil-1)((q^{k}-1)^{r}(q^{n}-1)+\frac{(q^{k}-1)^{r-1}(q^{n}-1)}{q-1})$. \qed

\begin{Remark}
In the case $n=5k$, the construction in {\rm \cite{FW2021}} yields a cyclic $(n,2k-2,k)_q$-CDC of size $2q^{k}\cdot \frac{q^{n}-1}{q-1}$,  the construction in {\rm \cite{NXG2022}} yields a cyclic $(n,2k-2,k)_q$-CDC of size $4q^{k}\cdot \frac{q^{n}-1}{q-1}$, while Theorem $\ref{Con1}$ gives a cyclic   $(n,2k-2,k)_q$-CDC of size  $2((q^{k}-1)^{2}(q^{n}-1)+\frac{(q^{k}-1)(q^{n}-1)}{q-1})$, which is larger. 
\end{Remark}

\section{The second construction}

In this section, we use irreducible polynomials and a variant of Sidon spaces in \cite{FW2021} to present a construction of cyclic CDCs. 

\begin{Lemma}
Let $q$ be a prime power, $k$ a positive integer, $c$ a nonzero element of  $\mathbb{F}_{q^{k}}$, and $\xi$  a primitive element of $\mathbb{F}_{q^{k}}$.   Let $A\subset \{0,1,\ldots,q^{k}-2\}$ be a set satisfying $c\xi^{i+j}\neq 1$ for any $i,j\in A$. Then $|A|=\lfloor \frac{q^{k}-2}{2}\rfloor$.
\end{Lemma}
\begin{proof}
Since $c\in \mathbb{F}^{*}_{q^{n}}$ and $\xi$ is a primitive element of $\mathbb{F}_{q^{k}}$, we can write $c=\xi^{m}$ for some $m\in \{0,1,\ldots,q^{k}-2\}$. The equation $c\xi^{i+j}\neq 1$ is simplified as 
\begin{equation*}
i+j+m\not\equiv 0 \pmod  {q^{k}-1}.
\end{equation*}
We construct a graph $G$ with vertex set $\{0,1,\ldots,q^{k}-2\}$ and $i \sim j$ if and only if $i+j+m\equiv 0 \pmod {q^{k}-1}.$ Since $i+j+m\not\equiv 0 \pmod {q^{k}-1}$ for all $i,j\in A$, then the maximum number of $A$ exactly equals to the independent number of $G$. It is easy to see that $\alpha=\lfloor \frac{q^{k}-2}{2}\rfloor$, thus $|A|= \lfloor \frac{q^{k}-2}{2}\rfloor$.
\end{proof}

 \begin{Construction}\label{Construction 2}
Let $k$ and $n$ be positive integers such that $k\geq 2$ and  $n=2rk$, where $r\geq 2$. Let $f(x)$ be an irreducible polynomial over $\mathbb{F}_{q^{k}}$ with degree $2r$. Let $\xi$ be a primitive element of $\mathbb{F}_{q^{k}}$ ,$\gamma$ a root of $f(x)$ and $A\subset \{0,1,\ldots,q^{k}-2\}$ satisfy $f(0)\xi^{i+j}\neq 1$ for any $i,j\in A$.  For $0\leq i_{m}\leq q^{k}-2$ $(1\leq m \leq r-1)$, $ i_{r}\in A$, $1\leq \ell \leq r-1$ and $0\leq b \leq q-2$, define 
\[
\begin{array}{l}
\mathcal{C}^{\ell}_{i_{1},i_{2},\ldots,i_{r},b}:= \{\alpha U^{\ell}_{i_{1},i_{2},\ldots,i_{r},b}\colon \alpha \in \mathbb{F}_{q^{n}}^{*}\},
\end{array}
\] where
\begin{equation}\label{Equ3}
U^{\ell}_{i_{1},i_{2},\ldots,i_{r},b}:=\{u+(u^{q}-\xi^{b}u)\xi^{i_{\ell}}\gamma^{\ell}+\sum\limits_{m=1,m\neq \ell}^ru\xi^{i_{m}}\gamma^{m}  \colon u\in \mathbb{F}_{q^{k}}\}
\end{equation}

Define
\[
\begin{array}{l}
\mathcal{C}:=\bigcup\limits_{0\leq i_{1},\ldots,i_{r-1}\leq q^k-2}\bigcup\limits_{ i_r\in A}\bigcup\limits_{ \ell=1}^{r-1}\bigcup\limits_{b=0}^{q-2}\mathcal{C}^{\ell}_{i_{1},i_{2},\ldots,i_{r},b}.
\end{array}
\] 
\end{Construction}

\begin{Theorem}\label{Con2}
The set $\mathcal{C}$ of subspaces of dimension $k$ defined in Construction $\ref{Construction 2}$ is a cyclic $(n,2k-2,k)_q$-CDC of size $\lfloor \frac{(r-1)(q^{k}-2)(q^{k}-1)^{r-1}(q^{n}-1)}{2}\rfloor$.
\end{Theorem}

In order to prove Theorem \ref{Con2}, we give a lemma.

\begin{Lemma}\label{NSidon}
The subspaces defined in $(\ref{Equ3})$ are Sidon spaces.
\end{Lemma}

\begin{proof}
We only prove that $U^{\ell}_{i_{1},i_{2},\ldots,i_{r},b}$ is a Sidon space. The other cases can be proved similarly.
Let
\[
\begin{array}{l}
\bar{u}=u+(u^{q}-\xi^{b}u)\xi^{i_{1}}\gamma+u\xi^{i_{2}}\gamma^{2}+\cdots+ u\xi^{i_{r-1}}\gamma^{(r-1)}+u\xi^{i_{r}}\gamma^{r},\smallskip \\
\bar{v}=v+(v^{q}-\xi^{b}v)\xi^{i_{1}}\gamma+v\xi^{i_{2}}\gamma^{2}+\cdots+ v\xi^{i_{r-1}}\gamma^{(r-1)}+v\xi^{i_{r}}\gamma^{r},\smallskip \\
\bar{s}=s+(s^{q}-\xi^{b}s)\xi^{i_{1}}\gamma+s\xi^{i_{2}}\gamma^{2}+\cdots+ s\xi^{i_{r-1}}\gamma^{(r-1)}+s\xi^{i_{r}}\gamma^{r},\smallskip \\
\bar{t}=t+(t^{q}-\xi^{b}t)\xi^{i_{1}}\gamma+t\xi^{i_{2}}\gamma^{2}+\cdots+ t\xi^{i_{r-1}}\gamma^{(r-1)}+t\xi^{i_{r}}\gamma^{r}
\end{array}
\]  
be distinct nonezero elements of $U^{\ell}_{i_{1},i_{2},\ldots,i_{r},b}$,  where $u$, $v$, $s$, $t$ are nonzero elements of $\mathbb{F}_{q^{k}}$. Suppose that $\bar{u}\bar{v}=\bar{s}\bar{t}$. We need to show that $\{\bar{u}\mathbb{F}_{q},\bar{v}\mathbb{F}_{q}\}=\{\bar{s}\mathbb{F}_{q},\bar{t}\mathbb{F}_{q}\}.$

Since $f(x)$ is an irreducible polynomial over $\mathbb{F}_{q^{k}}$ with degree $2r$ and the constant term $c=f(0)\neq 0$, and $\gamma$ is a root of $f(x)$, comparing the constant term after expansion of $\bar{u}\bar{v}=\bar{s}\bar{t}$ gives 
\begin{equation}\label{E1}
	uv(1-\xi^{2i_{r}}c)=st(1-\xi^{2i_{r}}c).
\end{equation}
According to the property of the set $A$, we have $(1-\xi^{2i_{r}}c)\neq 0$. It follows that $uv=st$.

Since $1$, $\gamma$, $\gamma^{j}$, $\ldots$, $\gamma^{2r-1}$ are linearly independent over $\mathbb{F}_{q^{k}}$,  comparing the coefficients of $\gamma^{2}$ in $\bar{u}\bar{v}=\bar{s}\bar{t}$ after expansion yields 
\begin{equation}\label{Eq2}
      u^{q}v+v^{q}u=s^{q}t+t^{q}s.
\end{equation}
Set $\frac{u}{s}=\lambda$. If $\lambda\in \mathbb{F}_{q}$, then $\frac{\bar{u}}{\bar{s}}, \frac{\bar{v}}{\bar{t}}\in \mathbb{F}_{q}$, thereby $\bar{u}\mathbb{F}_{q}=\bar{s}\mathbb{F}_{q}$ and $\bar{v}\mathbb{F}_{q}=\bar{t}\mathbb{F}_{q}$. If $\lambda \notin \mathbb{F}_{q}$, then after replacing $u,t$ with $\lambda s, \lambda v$ respectively, Equation (\ref{Eq2}) becomes the following: $$\lambda^{q}s^{q}v+\lambda v^{q}s=\lambda s^{q}v+\lambda^{q}v^{q}s.$$
Hence $s^{q}v=sv^{q}$ according to $\lambda^{q}\neq \lambda$, which
 implies that $\frac{s}{v}\in \mathbb{F}_{q}$,  $\bar{s}\mathbb{F}_{q}=\bar{v}\mathbb{F}_{q}$ and $\bar{u}\mathbb{F}_{q}=\bar{t}\mathbb{F}_{q}$.

So, the subspace $U^{\ell}_{i_{1},i_{2},\dots,i_{r}}$ is a Sidon space of dimension $k$ over $\mathbb{F}_{q}$.

We are in a position to prove Theorem \ref{Con2}.

\noindent {\bf{Proof of Theorem \ref{Con2}}:} 
Since each $U^{\ell}_{i_{1},i_{2},\ldots,i_{r},b}$ is a Sidon space by Lemma \ref{NSidon},  each $\mathcal{C}^{\ell}_{i_{1},i_{2},\ldots,i_{r},b}$ is a cyclic $(n,2k-2,k)_q$-CDC of size $\frac{q^{n}-1}{q-1}$ by Lemma \ref{sac}. To show that $\mathcal{C}$ has minimum distance $2k-2$, it remains to show that $\dim(\alpha U^{\ell}_{i_{1},i_{2},\ldots,i_{r},b}\cap U^{\ell'}_{i_{1}',i_{2}',\ldots,i_{r}',b'})\leq 1$ for any $\alpha \in \mathbb{F}_{q^{n}}^{*}$ and $(i_{1},i_{2},\ldots,i_{r},b,\ell)\neq (i_{1}',i_{2}',\ldots,i_{r}',b',\ell')$. By Lemma \ref{dist}, it suffices to show that for any nonzero elements $\bar{u},\bar{s}\in U^{\ell}_{i_{1},i_{2},\ldots,i_{r},b}$  and nonzero elements $\bar{v},\bar{t}\in U^{\ell'}_{i_{1}',i_{2}',\ldots,i_{r}',b'}$, the equality $\bar{u}\cdot \bar{v}=\bar{s}\cdot \bar{t}$ implies $\bar{u}\mathbb{F}_{q}=\bar{s}\mathbb{F}_{q}$ and $\bar{v}\mathbb{F}_{q}=\bar{t}\mathbb{F}_{q}$. 

Let $\bar{u}, \bar{s}, \bar{v}$ and $\bar{t}$ be of the following form:
\[
\begin{array}{l}
\bar{u}=u+(u^{q}-\xi^{b}u)\xi^{i_{\ell}}\gamma^{\ell}+\sum\limits_{m=1,m\neq \ell}^ru\xi^{i_{m}}\gamma^{m},\smallskip \\
\bar{v}=v+(v^{q}-\xi^{b'}v)\xi^{i_{\ell'}'}\gamma^{\ell'}+\sum\limits_{m=1,m\neq \ell}^rv\xi^{i_{m}'}\gamma^{m},\smallskip \\
\bar{s}=s+(s^{q}-\xi^{b}s)\xi^{i_{\ell}}\gamma^{\ell}+\sum\limits_{m=1,m\neq \ell}^rs\xi^{i_{m}}\gamma^{m},\smallskip \\
\bar{t}=t+(t^{q}-\xi^{b'}t)\xi^{i_{\ell'}'}\gamma^{\ell'}+\sum\limits_{m=1,m\neq \ell}^rt\xi^{i_{m}'}\gamma^{m},
\end{array}
\] 
where $u$, $v$, $s$, $t$ are nonzero elements of $\mathbb{F}_{q^{k}}$.   Simple computation shows that 
\begin{equation}\label{Equ4}
	\begin{split}
		&\quad \ \bar{u}\cdot \bar{v}\\
		&=uv+u(v^q-\xi^{b'}v)\xi^{i_{\ell'}'}\gamma^{\ell'}+\sum\limits_{m=1,m\neq \ell}^ruv\xi^{i_{m}'}\gamma^{m}+ (u^{q}-\xi^{b}u)v\xi^{i_{\ell}}\gamma^{\ell}\\
		&\quad +(u^{q}-\xi^{b}u)(v^q-\xi^{b'}v)\xi^{i_{\ell}+i_{\ell'}'}\gamma^{\ell+\ell'}+\sum\limits_{m=1,m\neq \ell}^r (u^{q}-\xi^{b}u)v\xi^{i_{\ell}+i_{m}'}\gamma^{m+\ell}\\
            &\quad +\sum\limits_{m=1,m\neq \ell}^ruv\xi^{i_{m}}\gamma^{m}+\sum\limits_{m=1,m\neq \ell}^ru(v^{q}-\xi^{b'}v)\xi^{i_{m}+i_{\ell'}'}\gamma^{\ell'+m}\\
            &\quad+\sum\limits_{m=1,m\neq \ell}^r\sum\limits_{x=1,x\neq \ell}^ruv\xi^{i_{m}+i_{x}'}\gamma^{m+x}.
	\end{split}
\end{equation}
Since $\gamma$ is a root of $f(x)$ with degree $2r$,  $\gamma^{2r}$ is an $\mathbb{F}_{q^{k}}$-linear combination of $1, \gamma,\ldots, \gamma^{2r-1}$. Comparing the constant term of $\bar{u}\cdot \bar{v}$ and $\bar{s}\cdot \bar{t}$ gives $uv(1-c\xi^{i_{r}+i_{r}'})=st(1-c\xi^{i_{r}+i_{r}'})$. According to the property of the set $A$, we have $(1-c\xi^{i_{r}+i_{r}'})\neq 0$. It follows that $uv=st$. 

We distinguish the following two cases

Case 1:  $\ell\neq \ell'$.

Without loss of generality, we can assume that $\ell<\ell'$. Since $1,\gamma,\ldots,\gamma^{2r-1}$ are linearly independent over $\mathbb{F}_{q^k}$, comparing coefficients of $\gamma^{\ell}$ in Equation (\ref{Equ4}) gives $u^qv=s^qt$. Because of  $uv=st$, we deduce that $(\frac{u}{s})^{q}=\frac{t}{v}=\frac{u}{s}$, thereby $\frac{u}{s}\in \mathbb{F}_{q}^{*}$. It follows that $\bar{u}\mathbb{F}_{q}=\bar{s}\mathbb{F}_{q}$ and $\bar{v}\mathbb{F}_{q}=\bar{t}\mathbb{F}_{q}$, as desired.

Case 2: $\ell=\ell'$.

Since $1,\gamma,\ldots,\gamma^{2r-1}$ are linearly independent over $\mathbb{F}_{q^k}$, comparing the coefficients of $\bar{u}\cdot \bar{v}$ and $\bar{s}\cdot \bar{t}$  gives the following:
\begin{equation}
uv=st,
\end{equation}
\begin{equation}
uv^{q}\xi^{i_{\ell}'}+u^{q}v\xi^{i_{\ell}}=st^{q}\xi^{i_{\ell}'}+s^{q}t\xi^{i_{\ell}},
\end{equation}
\begin{equation}
uv^{q}\xi^{b}+u^{q}v\xi^{b'}=st^{q}\xi^{b}+s^{q}t\xi^{b'},
\end{equation}
\begin{equation}
uv^{q}\xi^{i_{m}+i_{\ell}'}+u^{q}v\xi^{i_{m}'+i_{\ell}}=st^{q}\xi^{i_{m}+i_{\ell}'}+s^{q}t\xi^{i_{m}'+i_{\ell}}.
\end{equation}
for $1\leq m\neq \ell \leq r$.

Similar to the proof of Lemma \ref{CSC-C}, it has been checked that 
$\bar{u}\mathbb{F}_{q}=\bar{s}\mathbb{F}_{q}$, $\bar{v}\mathbb{F}_{q}=\bar{t}\mathbb{F}_{q}$, as desired.

To sum up, we have shown that $\mathcal{C}$ has minimum distance $2k-2$. In particular, we have shown that the $\lfloor \frac{(r-1)(q^{k}-2)(q^{k}-1)^{r-1}(q^{n}-1)}{2}\rfloor$ elements of $\mathcal{C}$ are distinct, so $|\mathcal{C}|=\lfloor \frac{(r-1)(q^{k}-2)(q^{k}-1)^{r-1}(q^{n}-1)}{2}\rfloor$. \qed
\end{proof}

\begin{Remark}
In the case of $n=4k$, the construction in {\rm \cite{FW2021}} yields a cyclic $(n,2k-2,k)_q$-CDC of size $q^{k}\cdot \frac{q^{4k}-1}{q-1}$, while Construction $\ref{Construction 2}$ yields a cyclic $(n,2k-2,k)_q$-CDC of size $\lfloor \frac{(q^{k}-2)(q^{k}-1)(q^{4k}-1)}{2}\rfloor$, which is larger.   By Lemma $\ref{UB}$, the upper bound on the size of a $(4k,2k-2,k)_q$-CDC is $\frac{(q^{4k}-1)(q^{4k-1}-1)}{(q^{k}-1)(q^{k-1}-1)}$.  When $k$ goes to infinity, the size of our CDC in Theorem $\ref{Con2}$  is within a factor of $\frac{1}{2}+o_{k}(1)$ of the upper bound in  Lemma $\ref{UB}$, where $o_{k}(1)$ stands for a term that tends to $0$ as $k$ tends to infinity. 
\end{Remark}

\section{Conclusion}
In this paper, we present two classes of Sidon spaces and two constructions of cyclic CDCs by taking the unions of Sidon spaces. Moreover, some cyclic CDCs constructed in this paper have more codewords than those in the literature, cf. Table I,  in particular, the size of our cyclic $(4k,2k-2,k)_q$-CDC in Theorem \ref{Con2} is within a factor of $\frac{1}{2}+o_{k}(1)$ of the sphere-packing bound as $k$ goes to infinity.

\end{document}